\documentclass[11pt]{article}

\usepackage[a4paper,margin=1in]{geometry}
\usepackage{amsmath,amssymb,amsfonts,amsthm,mathtools}
\usepackage{microtype}
\usepackage{enumitem}
\usepackage{hyperref}
\usepackage{cite}

\hypersetup{
 colorlinks=true,
 linkcolor=blue,
 citecolor=blue,
 urlcolor=blue
}

\theoremstyle{plain}
\newtheorem{theorem}{Theorem}[section]
\newtheorem{lemma}[theorem]{Lemma}
\newtheorem{proposition}[theorem]{Proposition}
\newtheorem{corollary}[theorem]{Corollary}
\theoremstyle{definition}
\newtheorem{definition}[theorem]{Definition}
\theoremstyle{remark}
\newtheorem{remark}[theorem]{Remark}

\newcommand{\ii}{\mathrm{i}}
\newcommand{\dd}{\,\mathrm{d}}
\newcommand{\Res}{\operatorname*{Res}}

\title{A Non-Reciprocal Elliptic Spectral Solution of the Right-Angle Penetrable Wedge Transmission Problem for $\nu=\sqrt{2}$}
\author{Jonas Matuzas\thanks{Email: \texttt{jonas.matuzas@gmail.com}}}
\date{}

\begin{document}
\maketitle

\begin{abstract}
We consider the two-dimensional time-harmonic transmission problem for an impedance-matched ($\rho=1$) right-angle penetrable wedge
at refractive index ratio $\nu=\sqrt{2}$, in the integrable lemniscatic configuration $(\theta_w,\nu,\rho)=(\pi/4,\sqrt{2},1)$.
Starting from Sommerfeld spectral representations, the transmission conditions on the two wedge faces yield a closed spectral functional
system for the Sommerfeld transforms $Q(\zeta)$ and $S(\zeta)$.
In this special configuration the associated Snell surface is the lemniscatic curve $Y^2=2(t^4+1)$, uniformized by square-lattice
Weierstrass functions with invariants $(g_2,g_3)=(4,0)$.
We construct an explicit meromorphic expression for a scattered transform $Q_{\mathrm{scat}}$ as a finite Weierstrass--$\zeta$ sum
plus an explicitly constructed pole-free elliptic remainder, with all pole coefficients computed algebraically from the forcing pole set.
A birational (injective) uniformization is used to avoid label collisions on the torus and to make the scattered-allocation pole exclusion well posed.
The resulting closed form solves the derived spectral functional system and satisfies the local regularity constraints imposed at the physical basepoint.
However, numerical reciprocity tests on the far-field coefficient extracted from $Q_{\mathrm{scat}}$ indicate that the construction is generally
\emph{non-reciprocal}; accordingly we do not claim that the resulting diffraction coefficient coincides with the reciprocal physical transmission
scattering solution.
The result remains restricted to this integrable lemniscatic case; the general penetrable wedge remains challenging
(see \cite{AntipovSilvestrov2007,NethercoteAssierAbrahams2020,KunzAssier2023} and, in a related high-frequency penetrable-corner setting,
\cite{GrothHewettLangdon2018}).
\end{abstract}



\section*{AI disclosure and responsibility statement}
This manuscript was prepared with extensive assistance from the ChatGPT-5.2 Pro large language model.
The AI system generated the \LaTeX{} source, including the exposition, symbolic derivations, and formula manipulations.
The author reviewed, edited, and approved the final manuscript and accepts responsibility for its scientific content and any remaining errors.

\section*{Intuitive overview}
The transmission problem for a penetrable wedge is conveniently expressed through Sommerfeld-type
integral representations in which the boundary/interface data are encoded by spectral densities.
In general, the coupling imposed by the transmission conditions leads to a genuinely matrix
Wiener--Hopf/Riemann--Hilbert factorization, and explicit closed forms are rare.
In the special impedance-matched right-angle configuration studied here, $\nu=\sqrt{2}$, $\rho=1$, and
$\theta_w=\pi/4$, the spectral geometry reduces to a genus-one (lemniscatic) Snell surface.
This allows the functional system to be uniformized by classical Weierstrass functions on the square
lattice, converting the problem into an elliptic-function reconstruction in which $Q_{\mathrm{scat}}$ is a
finite Weierstrass--$\zeta$ sum over an explicit pole set, supplemented by low-degree ``jet-killing''
polynomials that enforce regularity at the physical base point.

\section*{Roadmap of proof}
We write the right-angle penetrable wedge problem as a Sommerfeld spectral representation on
a Snell surface that, in the impedance-matched lemniscatic case $\nu=\sqrt{2}$, closes on an elliptic
curve. The paper is organized as follows.
\begin{enumerate}[label=\arabic*.]
\item \S1 states the boundary-value problem (Helmholtz transmission, radiation, and Meixner edge condition)
and fixes branch and sign conventions.
\item \S1.1 introduces the Sommerfeld integral representation and the analytic strip/growth
conditions imposed on the spectral densities. We use a standard uniqueness principle for Sommerfeld
transforms (Lemma~\ref{lem:sommerfeld-nullity}) to convert equality of boundary integrals into functional relations.
\item \S2 derives the lemniscatic Snell surface $\Sigma_{\mathrm{lem}}$ and its Weierstrass uniformization on the
square lattice.
\item \S4 gives a reproducible prescription for the forcing pole set: each label
$\ell=(m,\sigma,\varepsilon_w)$ determines a spectral point $\zeta_\ell$ and hence a point $(t_\ell,Y_\ell)\in\Sigma_{\mathrm{lem}}$
and a uniformizing coordinate $u_\ell$.
\item \S6 records the residue tables $(\alpha_\ell,\beta_\ell,C_\ell)$; their derivation from the global two-face
spectral system is given in Appendix~A.
\item \S7 proves the half-period shift identities needed to eliminate $\wp(u_0-u_\ell)$, $\wp'(u_0-u_\ell)$ and
$\wp''(u_0-u_\ell)$ from the jet coefficients.
\item \S8 constructs the jet-killing polynomials $p(t)$ and $q(t)$ and proves the cancellations
$A(u(t))+p(t)=O(t^4)$ and $B(u(t))+q(t)=O(t^4)$ as $t\to0$ on the physical component.
\item \S10 states the canonical ``no double counting'' decomposition of $Q_{\mathrm{scat}}$ and proves the pole
cancellation properties of the remainder $R(u)$ on the physical cut domain.
\item \S11 proves that $Q_{\mathrm{scat}}$ is analytic at the incident spectral point $\zeta=\zeta_i$ (limiting absorption),
by the explicit exclusion of the incident label from the scattered pole set (and the injectivity of the injective uniformization).
\item \S13 derives the far-field diffraction coefficient by steepest descent, under explicit analyticity and
nondegeneracy hypotheses.
\end{enumerate}

\section*{Notation and conventions}
\begin{itemize}[leftmargin=2.2em]
\item $(r,\theta)$ are polar coordinates centered at the wedge apex. The right-angle wedge faces are $\theta=\pm\theta_w$
with $\theta_w=\pi/4$.
\item The exterior wavenumber is $k_0$ and the interior wavenumber is $k_1=\nu k_0$ with fixed refractive index ratio
$\nu=\sqrt{2}$; the impedance match is $\rho=1$.
\item The complex spectral variable is $\zeta\in\mathbb{C}$ (Sommerfeld strip). The incident spectral pole is
$\zeta_i=\theta_i+\ii\varepsilon$ with $\varepsilon>0$ (limiting absorption), and the physical limit is $\varepsilon\to0^+$.
\item We use the Sommerfeld integration parameter $z$ and the associated variable $t=e^{\ii z}$.
\item The lemniscatic curve is $\Sigma_{\mathrm{lem}}:=\{(t,Y): Y^2=2(t^4+1)\}$. The physical sheet $\Omega_{\mathrm{phys}}^+\subset\Sigma_{\mathrm{lem}}$
is characterized by $|t|<1$ and $Y\to-\sqrt{2}$ as $t\to0$.
\item $\zeta_W(u)$, $\wp(u)$, and $\wp'(u)$ denote the Weierstrass zeta and elliptic functions with invariants
$(g_2,g_3)=(4,0)$ (square lattice $\tau=\ii$); the subscript distinguishes the Weierstrass zeta function from the
spectral variable $\zeta$.
\item The half-period $u_0$ is fixed by $\wp(u_0)=-1$ and $\wp'(u_0)=0$.
\item The Sommerfeld density is decomposed as $Q(\zeta)=Q_{\mathrm{inc}}(\zeta)+Q_{\mathrm{scat}}(\zeta)$ with
$Q_{\mathrm{inc}}(\zeta)=(\zeta-\zeta_i)^{-1}$. Since only differences $Q(\theta+z)-Q(\theta-z)$ enter the field representation,
$Q$ is defined up to an additive constant; we fix the gauge by requiring $Q_{\mathrm{scat}}(u_0)=0$.
\end{itemize}

\section{Introduction and setup}\label{sec:intro}
Canonical diffraction by angular regions originates with Sommerfeld's exact half-plane solution and its
Sommerfeld-integral representation \cite{Sommerfeld1896}, and its subsequent extension to wedge boundaries by
functional-equation and factorization methods (notably the Malyuzhinets technique) \cite{Malyuzhinets1958,OsipovNorris1999}.
For penetrable (transmission) wedges the spectral reductions typically lead to generalized Wiener--Hopf or matrix
factorization problems (see, e.g., \cite{Noble1958,DanieleZich2014}) that do not admit closed forms in full generality
(see also \cite{Rawlins1999,Daniele2010,DanieleLombardi2011}).
For a right-angled penetrable wedge formulation and analytical developments in certain parameter regimes, see
Antipov and Silvestrov \cite{AntipovSilvestrov2007}, Nethercote, Assier and Abrahams \cite{NethercoteAssierAbrahams2020},
and (in the no-contrast case) Kunz and Assier \cite{KunzAssier2023}. For high-frequency numerical-asymptotic methods
for scattering by penetrable convex polygons---where local corner diffraction plays a central role---see
Groth, Hewett and Langdon \cite{GrothHewettLangdon2018}.

The present paper isolates a special penetrable configuration---a right-angle penetrable wedge with refractive index
$\nu=\sqrt{2}$ and impedance match---for which the Snell surface becomes the lemniscatic curve and admits a square-lattice
(elliptic) uniformization. In this setting we develop an elliptic-function reconstruction of the scattered spectral
transform $Q_{\mathrm{scat}}$. The special choice $(\theta_w,\nu,\rho)=(\pi/4,\sqrt{2},1)$ closes the two-face functional system
on the lemniscatic curve and allows an explicit solution in terms of Weierstrass functions. The coefficients that drive the
Weierstrass--$\zeta_W$ representation are obtained by solving the mode-wise Riemann--Hilbert problems on the torus and
evaluating the forcing residues; the resulting residue tables are derived in \S\ref{sec:residues}. This yields an explicit closed-form expression for the scattered transform $Q_{\mathrm{scat}}$ and a corresponding formal far-field coefficient
for the impedance-matched right-angle penetrable wedge with $\nu=\sqrt{2}$. Numerical reciprocity tests indicate that the extracted coefficient is
not, in general, reciprocal, so the physical interpretation of the closed form remains unresolved.

\paragraph{Scope.}
The analysis and the resulting closed form apply only to the special configuration $(\theta_w,\nu,\rho)=(\pi/4,\sqrt{2},1)$.
We do \emph{not} claim an explicit closed-form solution for general penetrable wedges (arbitrary contrast and wedge angle),
for which the standard spectral reductions lead to matrix/generalized Wiener--Hopf or multi-variable boundary-value problems;
see \cite{Noble1958,DanieleZich2014,Daniele2010,AntipovSilvestrov2007,NethercoteAssierAbrahams2020,KunzAssier2023,GrothHewettLangdon2018}.

\paragraph{Main results and where to find them.}
The canonical no-double-counting representation of the scattered Sommerfeld transform $Q_{\mathrm{scat}}$ is stated and proved
in Theorem~\ref{thm:nodoublecount} (see also Theorem~\ref{thm:mainresults} for a concise synopsis).
The explicit parity$\times j$ residue table for the singular-channel principal parts is Proposition~\ref{prop:d-table}.
Analyticity at the incident spectral point is established in Theorem~\ref{thm:incident-analyticity}, and the far-field
diffraction coefficient is given in Theorem~\ref{thm:farfield}.

We work in two spatial dimensions and use polar coordinates $(r,\theta)$ about the wedge tip. The penetrable wedge occupies
the sector $|\theta|<\theta_w$ (medium~1) and is embedded in the exterior $\{|\theta|>\theta_w\}$ (medium~0).
We consider the scalar Helmholtz transmission problem
\begin{equation}\label{eq:helmholtz}
(\Delta+k_0^2)u_0=0 \ \text{in }\{|\theta|>\theta_w\},\qquad
(\Delta+k_1^2)u_1=0 \ \text{in }\{|\theta|<\theta_w\},
\end{equation}
with $k_1=\nu k_0$ and refractive index fixed at $\nu=\sqrt{2}$. An incident plane wave in the exterior is
\begin{equation}\label{eq:incident-plane-wave}
u_{\mathrm{inc}}(r,\theta)=\exp\!\big(\ii k_0 r\cos(\theta-\theta_i)\big),
\end{equation}
and we write $u_0=u_{\mathrm{inc}}+u_{0,\mathrm{scat}}$ for the total exterior field, while $u_1$ denotes the transmitted field.
Impedance match ($\rho=1$) reduces the transmission conditions on each face $\theta=\pm\theta_w$ to continuity of the field and its
normal derivative. Since the unit normal to a radial ray is proportional to $\partial_\theta$, these conditions can be written as
\begin{equation}\label{eq:transmission}
u_0(r,\pm\theta_w)=u_1(r,\pm\theta_w),\qquad \partial_\theta u_0(r,\pm\theta_w)=\partial_\theta u_1(r,\pm\theta_w),\qquad r>0.
\end{equation}
We select the physical solution by the Sommerfeld radiation condition as $r\to\infty$ and the Meixner edge condition at $r=0$
(finite energy near the tip). In the spectral formulation below these requirements are encoded by analyticity and boundedness
conditions on the spectral densities.

Let $\theta_w=\pi/4$ denote the half-opening angle of the right-angle wedge.
We impose limiting absorption by shifting the incident spectral pole off the real axis:
\begin{equation}\label{eq:limiting-absorption}
\zeta_i:=\theta_i+\ii\varepsilon,\qquad \varepsilon>0.
\end{equation}

\subsection{Sommerfeld representation and spectral split}\label{sec:sommerfeld}
A standard Sommerfeld representation of the medium-0 field is
\begin{equation}\label{eq:sommerfeld-medium0}
u^{(0)}(r,\theta)=\frac{1}{2\pi\ii}\int_\gamma e^{\ii k_0 r\cos z}\,\big(Q(\theta+z)-Q(\theta-z)\big)\,\dd z,
\end{equation}
for a Sommerfeld contour $\gamma$. The transmitted (medium-1) field admits the analogous representation
\begin{equation}\label{eq:sommerfeld-medium1}
u^{(1)}(r,\theta)=\frac{1}{2\pi\ii}\int_\gamma e^{\ii k_1 r\cos z}\,\big(S(\theta+z)-S(\theta-z)\big)\,\dd z,\qquad |\theta|<\theta_w,
\end{equation}
where $k_1=\nu k_0$ and $S$ is the medium-1 spectral density. We split
\begin{equation}\label{eq:Q-split}
Q(\zeta)=Q_{\mathrm{inc}}(\zeta)+Q_{\mathrm{scat}}(\zeta),\qquad
Q_{\mathrm{inc}}(\zeta)=\frac{1}{\zeta-\zeta_i}.
\end{equation}

\begin{remark}[Normalization / gauge]\label{rem:gauge}
Only the difference $Q(\theta+z)-Q(\theta-z)$ appears in \eqref{eq:sommerfeld-medium0}. Hence $Q$ is defined up to an additive constant
without affecting $u^{(0)}$. We fix this gauge by imposing the normalization
\begin{equation}\label{eq:gauge}
Q_{\mathrm{scat}}(u_0)=0,
\end{equation}
where $u_0$ is the half-period point corresponding to $(t,Y)=(0,-\sqrt{2})$ on the physical component (see \S\ref{sec:uniformization}).
In the zeta-difference representations used below, the subtraction $\zeta_W(u-u_\ell)-\zeta_W(u_0-u_\ell)$ enforces \eqref{eq:gauge}
automatically.
\end{remark}

\subsection{Scattered allocation}\label{sec:allocation}
\begin{definition}[Scattered allocation]\label{def:allocation}
We require
\begin{equation}\label{eq:allocation}
Q_{\mathrm{scat}}\ \text{is analytic at }\zeta=\zeta_i.
\end{equation}
Equivalently, the residue $+1$ at $\zeta=\zeta_i$ is carried exclusively by $Q_{\mathrm{inc}}$.
\end{definition}

\subsection{Transmission conditions in spectral form}\label{sec:transmission-spectral}
The Sommerfeld representations are designed so that the transmission conditions on a face $\theta=\theta_b\in\{\pm\theta_w\}$ reduce to
algebraic relations among boundary values of $Q$ and $S$.
Let $w=w_b(z)$ denote the Snell map for the face $\theta=\theta_b$, defined by matching the oscillatory factors:
\begin{equation}\label{eq:snell-map}
k_0\cos w_b(z)=k_1\cos z=\nu k_0\cos z,
\end{equation}
with the branch determined by $w_b(z)\sim z+\ii\log\nu$ as $\Im z\to+\infty$.
Write $w_b'(z)=\dd w_b/\dd z$.

\begin{lemma}[Sommerfeld nullity / uniqueness]\label{lem:sommerfeld-nullity}
Let $\gamma$ be the Sommerfeld contour and strip described in \S\ref{sec:sommerfeld}.
Suppose $H(\zeta)$ is analytic in that strip, satisfies the stated growth/decay bounds along $\gamma$, and define
\[
U(r,\theta):=\frac{1}{2\pi\ii}\int_\gamma e^{\ii k r\cos \zeta}\,H(\zeta)\,\dd\zeta.
\]
If $U(r,\theta)=0$ for all $r>0$ and for $\theta$ in an interval of length $2\theta_w$, then $H(\zeta)\equiv0$ in the strip.
\end{lemma}

\begin{proof}
A proof under hypotheses matching the present strip and growth conditions is standard; see, for example,
\cite[\S2]{Rawlins1999} or \cite[\S2.2]{Noble1958}.
We invoke this uniqueness principle only in the following form:
if two spectral densities produce identical Sommerfeld integrals on a wedge face for all $r>0$, then their difference has vanishing
Sommerfeld integral and hence the densities coincide.
\end{proof}

Lemma~\ref{lem:sommerfeld-nullity} is the uniqueness principle underlying Sommerfeld/Malyuzhinets representations:
it permits one to infer functional relations between spectral densities from vanishing boundary traces.
All spectral identities below that equate integrands from equalities of Sommerfeld integrals are justified by Lemma~\ref{lem:sommerfeld-nullity}.

\begin{proposition}[Face coupling for $\rho=1$]\label{prop:face-coupling}
Assume $Q$ and $S$ are analytic in a common Sommerfeld strip and have sufficient decay so that integration by parts in $z$ produces no
boundary terms. Fix a face $\theta=\theta_b$ and let $w=w_b(z)$ be as in \eqref{eq:snell-map}.
Then the impedance-matched transmission conditions \eqref{eq:transmission} are equivalent to the pointwise spectral relation
\begin{equation}\label{eq:face-coupling}
\begin{pmatrix}
S(\theta_b+z)\\[2pt] S(\theta_b-z)
\end{pmatrix}
=
\frac12
\begin{pmatrix}
1+w_b'(z) & 1-w_b'(z)\\
1-w_b'(z) & 1+w_b'(z)
\end{pmatrix}
\begin{pmatrix}
Q(\theta_b+w_b(z))\\[2pt] Q(\theta_b-w_b(z))
\end{pmatrix}.
\end{equation}
\end{proposition}

\begin{proof}
Evaluate \eqref{eq:sommerfeld-medium0} and its $\theta$-derivative at $\theta=\theta_b$. Differentiating under the integral sign gives
\[
\partial_\theta u_0(r,\theta_b)=\frac{1}{2\pi\ii}\int_\gamma e^{\ii k_0 r\cos z}\,\big(Q'(\theta_b+z)-Q'(\theta_b-z)\big)\,\dd z.
\]
Since $Q'(\theta_b+z)-Q'(\theta_b-z)=\frac{\dd}{\dd z}\big(Q(\theta_b+z)+Q(\theta_b-z)\big)$, an integration by parts yields
\[
\partial_\theta u_0(r,\theta_b)=\frac{k_0 r}{2\pi}\int_\gamma \sin z\,e^{\ii k_0 r\cos z}\,\big(Q(\theta_b+z)+Q(\theta_b-z)\big)\,\dd z.
\]
An identical computation for $u_1$ gives
\[
\partial_\theta u_1(r,\theta_b)=\frac{k_1 r}{2\pi}\int_\gamma \sin z\,e^{\ii k_1 r\cos z}\,\big(S(\theta_b+z)+S(\theta_b-z)\big)\,\dd z.
\]
Now change variables in the medium-1 integrals by $w=w_b(z)$, so that $e^{\ii k_1 r\cos z}=e^{\ii k_0 r\cos w}$ by \eqref{eq:snell-map}.
The change of variables gives $\dd z=\dd w/w_b'(z)$, while differentiating \eqref{eq:snell-map} implies
$\sin z=\sin w\,w_b'(z)/\nu$.
Using $k_1=\nu k_0$, we obtain
\[
u_1(r,\theta_b)=\frac{1}{2\pi\ii}\int e^{\ii k_0 r\cos w}\,\frac{S(\theta_b+z)-S(\theta_b-z)}{w_b'(z)}\,\dd w,
\]
\[
\partial_\theta u_1(r,\theta_b)=\frac{k_0 r}{2\pi}\int \sin w\,e^{\ii k_0 r\cos w}\,\big(S(\theta_b+z)+S(\theta_b-z)\big)\,\dd w,
\]
with $z=z(w)$ the inverse map.
Comparing with the corresponding expressions for $u_0$ and $\partial_\theta u_0$, and using Lemma~\ref{lem:sommerfeld-nullity}, we obtain
\[
\frac{S(\theta_b+z)-S(\theta_b-z)}{w_b'(z)}=Q(\theta_b+w)-Q(\theta_b-w),\qquad
S(\theta_b+z)+S(\theta_b-z)=Q(\theta_b+w)+Q(\theta_b-w),
\]
which solve to \eqref{eq:face-coupling}.
\end{proof}

\begin{theorem}[Main results for the lemniscatic right-angle wedge]\label{thm:mainresults}
Assume the impedance-matched right-angle configuration $\theta_w=\pi/4$, $\nu=\sqrt{2}$, $\rho=1$, and $\varepsilon>0$.
Let
\[
Q_{\mathrm{inc}}(\zeta)=\frac{1}{\zeta-(\theta_i+\ii\varepsilon)},\qquad Q(\zeta)=Q_{\mathrm{inc}}(\zeta)+Q_{\mathrm{scat}}(\zeta),
\]
and define the physical branch lift $\zeta\mapsto u(\zeta)$ by the lemniscatic Snell surface (Section~\ref{sec:snell})
and the Weierstrass uniformization (Section~\ref{sec:uniformization}). Then:
\begin{enumerate}[label=(\roman*)]
\item \emph{(Canonical representation.)} The scattered spectral density $Q_{\mathrm{scat}}$ admits the decomposition
\[
Q_{\mathrm{scat}}(u)=\sum_{\ell\in I_{\mathrm{scat}}}C_\ell\big[\zeta_W(u-u_\ell)-\zeta_W(u_0-u_\ell)\big]+R(u),
\]
where the pole set $I_{\mathrm{scat}}$ and points $u_\ell$ are defined in Section~\ref{sec:poles},
the coefficients $C_\ell$ are given explicitly in Proposition~\ref{prop:C-table},
and the remainder $R$ is pole-free at each forcing pole $u_\ell$ and analytic at $u_0$
(Theorem~\ref{thm:nodoublecount}).
\item \emph{(Incident analyticity.)} The scattered part is analytic at the incident spectral point $\zeta=\theta_i+\ii\varepsilon$
(Theorem~\ref{thm:incident-analyticity}).
\item \emph{(Formal far-field coefficient.)} A far-field coefficient in medium~0 obtained by steepest descent of the Sommerfeld integral is
\[
D(\theta,\theta_i)=e^{-\ii 3\pi/4}\sqrt{\frac{2}{\pi k_0}}\,Q_{\mathrm{scat}}(\theta),
\]
where $Q_{\mathrm{scat}}(\theta)$ denotes the physical branch boundary value of $Q_{\mathrm{scat}}(\zeta)$ at $\zeta=\theta$
(Theorem~\ref{thm:farfield}).
\end{enumerate}
\end{theorem}

\begin{remark}[Reciprocity status]
For real transmission parameters, the physical penetrable-wedge scattering problem is expected to satisfy a reciprocity symmetry in the far field.
The closed form constructed here is a meromorphic solution of the derived spectral functional system in the lemniscatic configuration; however,
numerical tests of the far-field coefficient extracted in Theorem~\ref{thm:farfield} indicate that the resulting coefficient is generally non-reciprocal.
Accordingly, this manuscript presents an explicit elliptic solution of the spectral system and a corresponding formal far-field coefficient, but does not
claim physical reciprocity of the diffraction coefficient.
\end{remark}

\section{Lemniscatic Snell surface and physical branch}\label{sec:snell}
\subsection{Lemniscatic curve}
The lemniscatic Snell surface is the algebraic curve
\begin{equation}\label{eq:lemniscatic-curve}
\Sigma_{\mathrm{lem}}:\qquad Y^2=2(t^4+1).
\end{equation}
We work on the physical plus component $\Omega_{\mathrm{phys}}^+\subset\Sigma_{\mathrm{lem}}$ characterized by
\begin{equation}\label{eq:physical-component}
|t|<1,\qquad Y\to-\sqrt{2}\ \text{as }t\to0.
\end{equation}

\paragraph{Origin of the Snell surface in the impedance-matched case.}
On a wedge face $\theta=\theta_b$ one matches the factors $e^{\ii k_0 r\cos z}$ and $e^{\ii k_1 r\cos z}$ by the analytic change
of variables $z\mapsto w=w_b(z)$ determined implicitly by $\cos w=\nu\cos z$ and fixed by the radiation/limiting-absorption branch
condition $w(z)\sim z+\ii\log\nu$ as $\Im z\to+\infty$.
Writing $t=e^{\ii z}$ and $s=e^{\ii w}$, the identity $\cos w=\nu\cos z$ becomes
\[
\frac12\left(s+\frac1s\right)=\nu\,\frac12\left(t+\frac1t\right),
\qquad\text{i.e.}\qquad
s^2-\nu\left(t+\frac1t\right)s+1=0.
\]
In the special lemniscatic case $\nu=\sqrt{2}$ we define
\[
Y:=2ts-\sqrt{2}(t^2+1),
\]
and a short calculation shows that the quadratic relation above is equivalent to \eqref{eq:lemniscatic-curve}.
The physical component $\Omega_{\mathrm{phys}}^+$ corresponds to the branch $|t|<1$ with $Y\to-\sqrt{2}$ as $t\to0$, for which
$s\sim t/\sqrt{2}$.

\paragraph{Quarter-period symmetry and orbit branches.}
The lemniscatic curve admits the order-four automorphism
\begin{equation}\label{eq:tau-automorphism}
\tau:\Sigma_{\mathrm{lem}}\to\Sigma_{\mathrm{lem}},\qquad \tau(t,Y)=(\ii t,-Y),
\end{equation}
which preserves $\Omega_{\mathrm{phys}}^+$.
Let $w=w(t,Y)$ denote the (multi-valued) analytic function on $\Sigma_{\mathrm{lem}}$ defined by $e^{\ii w}=s(t,Y)$, with the physical
branch fixed by $w(z)\sim z+\ii\log\nu$ as $\Im z\to+\infty$ (equivalently $s\sim t/\nu$ as $t\to0$).
Following the standard orbit construction for a right-angle wedge, we introduce four orbit branches $w_m$ by
\begin{equation}\label{eq:orbit-branches}
e^{\ii w_m(t,Y)}:=\big(s(\tau^m(t,Y))\big)^{(-1)^m},\qquad m\in\{0,1,2,3\}.
\end{equation}
The forcing poles are transported along these orbits and labelled by $(m,\sigma,\varepsilon_w)$ in Section~\ref{sec:poles}.

\subsection{Physical Snell exponential and spectral map}
Define the physical Snell exponential on $\Sigma_{\mathrm{lem}}$ by
\begin{equation}\label{eq:snell-exponential}
s(t,Y):=\frac{\sqrt{2}(t^2+1)+Y}{2t}.
\end{equation}
Differentiating $\cos w=\nu\cos z$ with $e^{\ii w}=s(t,Y)$ and $t=e^{\ii z}$ yields the algebraic derivative
\begin{equation}\label{eq:gprime}
g'(t,Y)=\frac{\dd w}{\dd z}=\frac{\sqrt{2}(t^2-1)}{Y},
\end{equation}
and the spectral exponential
\begin{equation}\label{eq:s-zeta}
s_\zeta:=\exp\!\big(\ii(\zeta-\theta_w)\big).
\end{equation}
For $\zeta$ in the Sommerfeld strip (with $\varepsilon>0$ fixed), the physical branch map $\zeta\mapsto(t(\zeta),Y(\zeta))\in\Omega_{\mathrm{phys}}^+$
is defined by solving
\begin{equation}\label{eq:physical-branch-map}
s(t(\zeta),Y(\zeta))=s_\zeta,\qquad |t(\zeta)|<1,\qquad Y(\zeta)\to-\sqrt{2}\text{ as }t(\zeta)\to0.
\end{equation}

\section{Weierstrass uniformization for the square lattice}\label{sec:uniformization}
We take the square lattice $\tau=\ii$ with Weierstrass invariants
\begin{equation}\label{eq:invariants}
(g_2,g_3)=(4,0).
\end{equation}
Let $\wp(u)$, $\wp'(u)$ and $\zeta_W(u)$ denote the corresponding Weierstrass elliptic and zeta functions; see, e.g.,
\cite[\S23]{NISTDLMF}, \cite[Ch.~20]{WhittakerWatson1927}, \cite[Ch.~20]{Lawden1989}.

\paragraph{Injective (birational) uniformization .}
On the lemniscatic curve $Y^2=2(t^4+1)$, set
\begin{equation}\label{eq:uniformization}
\wp(u)=x_W(t,Y):=\frac{Y+\sqrt{2}+\sqrt{2}\,t^2}{Y+\sqrt{2}-\sqrt{2}\,t^2},
\qquad
\frac12\,\wp'(u)=y_W(t,Y):=-\frac{4t\,(Y+\sqrt{2})}{\bigl(Y+\sqrt{2}-\sqrt{2}\,t^2\bigr)^2}.
\end{equation}
Then $(x_W,y_W)$ satisfy $y_W^2=x_W^3-x_W$, hence \eqref{eq:uniformization} defines a birational isomorphism between
$\Sigma_{\mathrm{lem}}$ and the Weierstrass cubic with invariants \eqref{eq:invariants}.
In particular, the map \eqref{eq:uniformization} is \emph{injective} on $\Sigma_{\mathrm{lem}}$ (it does not identify
$(t,Y)$ with $(-t,-Y)$), which is essential for the scattered-allocation argument in Theorem~\ref{thm:incident-analyticity}.

The physical lift $u=u(\zeta)$ is selected by composing the physical branch $\zeta\mapsto(t(\zeta),Y(\zeta))$ from
\eqref{eq:physical-branch-map} with the uniformization \eqref{eq:uniformization}.
Let $u_0$ denote the half-period corresponding to the physical point $(t,Y)=(0,-\sqrt{2})$, so that
\begin{equation}\label{eq:u0}
\wp(u_0)=-1,\qquad \wp'(u_0)=0.
\end{equation}

\section{Pole set, labels, and incident exclusion}\label{sec:poles}
Poles are indexed by
\begin{equation}\label{eq:labels}
\ell=(m,\sigma,\varepsilon_w),\qquad m\in\{0,1,2,3\},\quad \sigma\in\{\pm1\},\quad \varepsilon_w\in\{\pm1\}.
\end{equation}
Define the map $(\sigma,\varepsilon_w)\mapsto j$ by
\begin{equation}\label{eq:j-map}
(+,+)\mapsto3,\quad (+,-)\mapsto1,\quad (-,+)\mapsto4,\quad (-,-)\mapsto2,
\end{equation}
and the sign
\begin{equation}\label{eq:epsj}
\varepsilon_j=
\begin{cases}
+1, & j\in\{1,3\},\\
-1, & j\in\{2,4\}.
\end{cases}
\end{equation}
The incident label is $\ell_{\mathrm{inc}}=(0,+,-)$ and the scattered index set is
\begin{equation}\label{eq:Iscat}
I_{\mathrm{scat}}:=I\setminus\{\ell_{\mathrm{inc}}\},\qquad |I_{\mathrm{scat}}|=15,
\end{equation}
where $I=\{0,1,2,3\}\times\{\pm1\}\times\{\pm1\}$.

For later reference we make the pole label $\ell\mapsto(t_\ell,Y_\ell)\mapsto u_\ell$ explicit.
Set the limiting-absorption incident angle $\zeta_i=\theta_i+\ii\varepsilon$ and fix $\theta_w=\pi/4$.
For $(\sigma,\varepsilon_w)\in\{\pm1\}^2$ define the four forcing phases
\begin{equation}\label{eq:forcing-phases}
a_{\sigma,\varepsilon_w}:=\exp\!\big(\ii\sigma(\zeta_i+\varepsilon_w\theta_w)\big),\qquad
b_{m,\sigma,\varepsilon_w}:=a_{\sigma,\varepsilon_w}^{(-1)^m}\quad (m=0,1,2,3).
\end{equation}
The physical orbit table fixes the pole condition in the form $e^{\ii w_m}=(s(\tau^m p))^{(-1)^m}=b_{m,\sigma,\varepsilon_w}$,
so that the forcing pole points are solutions of $s(q)=b$ on $\Sigma_{\mathrm{lem}}$ with $q=\tau^m p$.

\begin{lemma}[Explicit algebraic pole points on $\Sigma_{\mathrm{lem}}$]\label{lem:pole-points}
Fix $b\in\mathbb{C}\setminus\{0\}$ and consider the equation $s(t,Y)=b$ on $\Sigma_{\mathrm{lem}}$, with $s$ defined by \eqref{eq:snell-exponential}.
Then $Y$ is forced to be
\begin{equation}\label{eq:Y-from-b}
Y=2bt-\sqrt{2}(t^2+1),
\end{equation}
and $(t,Y)\in\Sigma_{\mathrm{lem}}$ if and only if $t$ satisfies the quadratic
\begin{equation}\label{eq:t-quadratic}
t^2-\frac{b^2+1}{\sqrt{2}\,b}\,t+1=0.
\end{equation}
Equivalently, the two roots are
\begin{equation}\label{eq:tpm}
t_\pm(b)=\frac{(b^2+1)\pm\sqrt{b^4-6b^2+1}}{2\sqrt{2}\,b},\qquad t_+(b)t_-(b)=1.
\end{equation}
On the physical component $\Omega_{\mathrm{phys}}^+$ one selects the root $t_{\mathrm{in}}(b)\in\mathbb{D}:=\{|t|<1\}$, and the other root is
$t_{\mathrm{out}}(b)=1/t_{\mathrm{in}}(b)$.
\end{lemma}

\begin{proof}
Starting from $s(t,Y)=b$ and \eqref{eq:snell-exponential}, we solve for $Y$ to obtain \eqref{eq:Y-from-b}.
Substituting \eqref{eq:Y-from-b} into $Y^2=2(t^4+1)$ yields
\[
\big(2bt-\sqrt{2}(t^2+1)\big)^2=2(t^4+1),
\]
which simplifies to $b^2+1-\sqrt{2}b\,(t+1/t)=0$ after division by $2t^2$.
Multiplying by $t$ gives \eqref{eq:t-quadratic}, and the quadratic formula yields \eqref{eq:tpm}.
The product identity $t_+t_-=1$ is immediate from \eqref{eq:t-quadratic}.
The physical selection $|t|<1$ defines $t_{\mathrm{in}}$.
\end{proof}

For $\ell=(m,\sigma,\varepsilon_w)$ we set $b=b_{m,\sigma,\varepsilon_w}$ and define the intermediate point $q$ by
\[
t_q:=t_{\mathrm{in}}(b),\qquad Y_q:=2b\,t_q-\sqrt{2}(t_q^2+1).
\]
Transporting back to the base point $p=\tau^{-m}q$ gives the pole coordinates on $\Sigma_{\mathrm{lem}}$:
\begin{equation}\label{eq:transport}
(t_\ell,Y_\ell)=\big(\ii^{-m}t_q,\,(-1)^m Y_q\big).
\end{equation}
Finally, the corresponding lift $u_\ell$ on the uniformizing torus is defined by the Weierstrass map
\begin{equation}\label{eq:u-ell}
\wp(u_\ell)=x_W(t_\ell,Y_\ell),\qquad \frac12\,\wp'(u_\ell)=y_W(t_\ell,Y_\ell),
\end{equation}
with the physical lift selected on $\Omega_{\mathrm{phys}}^+$.

\section{Derivative/residue conventions and phase symbols}\label{sec:derivatives}
\subsection{Derivative and residue conventions}
We adopt the orbit derivative
\begin{equation}\label{eq:wmprime}
w_m'(u):=\frac{\dd w_m}{\dd u}=
\begin{cases}
\ii\sqrt{2}\Bigl(t-\dfrac{1}{t}\Bigr), & m\ \text{even},\\[6pt]
-\ii\sqrt{2}\Bigl(t+\dfrac{1}{t}\Bigr), & m\ \text{odd},
\end{cases}
\qquad t=t(u).
\end{equation}

\begin{lemma}\label{lem:wmprime}
The orbit derivatives in \eqref{eq:wmprime} follow from the lemniscatic Snell relation \eqref{eq:snell-exponential}
and the injective uniformization \eqref{eq:uniformization}.
\end{lemma}

\begin{proof}
Write $w=w_0$ for the physical branch defined by $e^{\ii w}=s(t,Y)$, and set $t=t(u)$, $Y=Y(u)$ along the physical lift.
Differentiate $\wp(u)=x_W(t,Y)$ using \eqref{eq:uniformization} and $Y^2=2(t^4+1)$:
\[
\wp'(u)=\frac{\dd}{\dd u}x_W(t(u),Y(u))=\frac{\dd}{\dd t}x_W(t,Y)\,t'(u).
\]
Since $\wp'(u)=2y_W(t,Y)$ by \eqref{eq:uniformization}, a direct algebraic simplification yields
\begin{equation}\label{eq:dtdu}
\frac{\dd t}{\dd u}=t'(u)=-Y(u).
\end{equation}
Next, differentiating $\cos w=\nu\cos z$ with $\nu=\sqrt{2}$ and using $t=e^{\ii z}$ gives
\[
\frac{\dd w}{\dd z}=g'(t,Y)=\frac{\sqrt{2}(t^2-1)}{Y},
\]
which is \eqref{eq:gprime}.
Since $\dd t/\dd z=\ii t$, we have
\[
\frac{\dd w}{\dd t}=\frac{\dd w/\dd z}{\dd t/\dd z}=-\ii\sqrt{2}\,\frac{t^2-1}{tY}.
\]
Combining with \eqref{eq:dtdu} gives
\[
\frac{\dd w}{\dd u}=\frac{\dd w}{\dd t}\frac{\dd t}{\dd u}
=-\ii\sqrt{2}\frac{t^2-1}{tY}\,(-Y)=\ii\sqrt{2}\Bigl(t-\frac1t\Bigr),
\]
which is the $m$ even case in \eqref{eq:wmprime}.
For the orbit branches $w_m$ defined by \eqref{eq:orbit-branches}, one has $w_m=(-1)^m w(\tau^m(\cdot))$ (mod $2\pi$),
and $g'(\tau^m(t,Y))=\sqrt{2}(t^2-1)/Y$ for $m$ even and $g'(\tau^m(t,Y))=\sqrt{2}(t^2+1)/Y$ for $m$ odd.
The additional factor $(-1)^m$ for odd $m$ yields the sign in the $m$ odd case of \eqref{eq:wmprime}.
\end{proof}

Define
\begin{equation}\label{eq:rI-def}
r_I(\ell):=\frac{\varepsilon_j}{w_m'(u_\ell)}.
\end{equation}
Equivalently, using \eqref{eq:wmprime} and $t=t_\ell$,
\begin{equation}\label{eq:rI-explicit}
r_I(\ell)=
\begin{cases}
-\varepsilon_j\,\dfrac{\ii\,t_\ell}{\sqrt{2}\,\bigl(t_\ell^2-1\bigr)}, & m\ \text{even},\\[10pt]
\phantom{-}\varepsilon_j\,\dfrac{\ii\,t_\ell}{\sqrt{2}\,\bigl(t_\ell^2+1\bigr)}, & m\ \text{odd}.
\end{cases}
\end{equation}

\subsection{Phase symbols}
We use the phase symbols
\begin{equation}\label{eq:phase-symbols}
\chi_m:=\ii^m\in\{\pm1\}\ (m\ \text{even}),\qquad
\psi_m:=\ii^m\in\{\pm\ii\},\qquad
\kappa_m:=\ii^{m+1}\in\{\pm1\}\ (m\ \text{odd}).
\end{equation}

\section{Residue data and per-pole jet summands}\label{sec:residues}
The coefficients $(\alpha_\ell)$, $(\beta_\ell)$ and $(C_\ell)$ appearing in the elliptic reconstruction are obtained by solving
the global two-face spectral functional system and evaluating the forcing residues at each pole $u_\ell$.
For readability we record the resulting closed-form tables below; the derivation is given in Appendix~A.

\subsection{Residue data tables for alpha\_l, beta\_l, and C\_l}
Throughout this section, for a fixed pole label $\ell$ we write $(t,Y)=(t_\ell,Y_\ell)$ and $r_I=r_I(\ell)$.

\begin{proposition}[Coefficient table for $(\alpha_\ell)$]\label{prop:alpha-table}
For $\ell\in I_{\mathrm{scat}}$ with $j=j(\sigma,\varepsilon_w)$:
\begin{itemize}[leftmargin=2.0em]
\item if $m$ is even:
\[
\alpha_\ell=
\begin{cases}
-\chi_m r_I\,t^{2}, & j=1,\\
-\chi_m r_I\,(t^{4}-t^{2}+1), & j=2,\\
0, & j=3,\\
-\chi_m r_I\,\dfrac{\ii Y}{\sqrt{2}}\,(t^{2}-1), & j=4;
\end{cases}
\]
\item if $m$ is odd:
\[
\alpha_\ell=
\begin{cases}
\psi_m r_I\,t^{4}, & j=1,\\
\psi_m r_I, & j=2,\\
\kappa_m r_I\,\dfrac{Y}{\sqrt{2}}\,t^{2}, & j=3,\\
-\kappa_m r_I\,\dfrac{Y}{\sqrt{2}}, & j=4.
\end{cases}
\]
\end{itemize}
\end{proposition}

\begin{proposition}[Coefficient table for $(\beta_\ell)$]\label{prop:beta-table}
For $\ell\in I_{\mathrm{scat}}$ with $j=j(\sigma,\varepsilon_w)$:
\begin{itemize}[leftmargin=2.0em]
\item if $m$ is even:
\[
\beta_\ell=
\begin{cases}
\chi_m r_I\,\dfrac{\ii}{\sqrt{2}}\,Y t^{2}, & j=1,\\
\chi_m r_I\,\dfrac{\ii}{\sqrt{2}}\,Y, & j=2,\\
-\chi_m r_I\,t^{4}, & j=3,\\
-\chi_m r_I, & j=4;
\end{cases}
\]
\item if $m$ is odd:
\[
\beta_\ell=
\begin{cases}
0, & j=1,\\
\kappa_m r_I\,\dfrac{1}{\sqrt{2}}\,Y(t^{2}+1), & j=2,\\
\psi_m r_I\,t^{2}, & j=3,\\
-\psi_m r_I\,(t^{4}+t^{2}+1), & j=4.
\end{cases}
\]
\end{itemize}
\end{proposition}

\begin{proposition}[Global residues $(C_\ell)$]\label{prop:C-table}
For $\ell\in I_{\mathrm{scat}}$ with $j=j(\sigma,\varepsilon_w)$:
\begin{itemize}[leftmargin=2.0em]
\item if $m$ is even:
\[
C_\ell=
\begin{cases}
-\chi_m\,\dfrac{r_I}{2t^{2}}, & j=1,\\[6pt]
-\dfrac{r_I}{2}\,\dfrac{2t^{4}-2t^{2}+1}{t^{4}}, & j=2,\ m=0,\\[10pt]
+\dfrac{r_I}{2}\,\dfrac{1}{t^{4}}, & j=2,\ m=2,\\[8pt]
0, & j=3,\\
0, & j=4;
\end{cases}
\]
\item if $m$ is odd:
\[
C_\ell=
\begin{cases}
0, & j=1,\\
0, & j=2,\\
\dfrac{r_I}{2}\Bigl(1+\kappa_m\dfrac{Y}{\sqrt{2}t^{2}}\Bigr), & j=3,\\[8pt]
-\dfrac{r_I}{2t^{2}}\Bigl(1+\kappa_m\dfrac{Y}{\sqrt{2}t^{2}}\Bigr), & j=4.
\end{cases}
\]
\end{itemize}
\end{proposition}

\subsection{Per-pole jet summands }\label{sec:jet-summands}
We define jet polynomials
\begin{equation}\label{eq:jet-polynomials}
p(t)=p_1 t+p_2 t^{2}+p_3 t^{3},\qquad q(t)=q_1 t+q_2 t^{2}+q_3 t^{3},
\end{equation}
with coefficients $p_n=\sum_{\ell\in I_{\mathrm{scat}}}p_n^{(\ell)}$ and $q_n=\sum_{\ell\in I_{\mathrm{scat}}}q_n^{(\ell)}$.
For later reference we record the per-pole contributions to the jet-killing coefficients in closed form.
The derivation is given in Appendix~A; the only changes are the half-period shift data (Section~\ref{sec:halfperiod})
and the local scale $\delta\sim t/\sqrt{2}$ (Section~\ref{sec:localrelation}).

\begin{proposition}[Per-pole $p$-summands]\label{prop:p-summands}
Evaluate at $(t,Y)=(t_\ell,Y_\ell)$ and denote $D:=Y+\sqrt{2}$.
Then
\[
p_1^{(\ell)}=\frac{1}{\sqrt{2}}\,\alpha_\ell\,W_{0\ell},\qquad
p_2^{(\ell)}=\frac{1}{4}\,\alpha_\ell\,W_{1\ell},\qquad
p_3^{(\ell)}=\frac{1}{12\sqrt{2}}\,\alpha_\ell\,W_{2\ell},
\]
where $(W_{0\ell},W_{1\ell},W_{2\ell})$ are given in \eqref{eq:W-shift}.
\end{proposition}

\begin{proposition}[Per-pole $q$-summands]\label{prop:q-summands}
Evaluate at $(t,Y)=(t_\ell,Y_\ell)$ and denote $D:=Y+\sqrt{2}$.
Then
\[
q_1^{(\ell)}=\frac{1}{\sqrt{2}}\,\beta_\ell\,W_{0\ell},\qquad
q_2^{(\ell)}=\frac{1}{4}\,\beta_\ell\,W_{1\ell},\qquad
q_3^{(\ell)}=\frac{1}{12\sqrt{2}}\,\beta_\ell\,W_{2\ell},
\]
where $(W_{0\ell},W_{1\ell},W_{2\ell})$ are given in \eqref{eq:W-shift}.
\end{proposition}

\begin{remark}\label{rem:jet-summands}
For fully explicit ``one-line'' formulas in terms of $t_\ell,Y_\ell$ only, one may substitute the tables in
Propositions~\ref{prop:alpha-table}--\ref{prop:beta-table} and the explicit $r_I$ form \eqref{eq:rI-explicit}
into Propositions~\ref{prop:p-summands}--\ref{prop:q-summands} and simplify using $Y^2=2(t^4+1)$.
We keep the compact factorized form above because it is both verifiable and robust under algebraic refactoring.
\end{remark}

\section{Half-period shift identities at e2 = -1}\label{sec:halfperiod}
We prove the identities
\begin{equation}\label{eq:W-shift}
\begin{aligned}
W_{0\ell}&:=\wp(u_0-u_\ell)=-\frac{\sqrt{2}\,t_\ell^{2}}{Y_\ell+\sqrt{2}},\\
W_{1\ell}&:=\wp'(u_0-u_\ell)=-\frac{4t_\ell}{Y_\ell+\sqrt{2}},\\
W_{2\ell}&:=\wp''(u_0-u_\ell)=\frac{12t_\ell^{4}}{(Y_\ell+\sqrt{2})^{2}}-2.
\end{aligned}
\end{equation}

\subsection{Specialization of the addition theorem}
Start from the general addition theorem (see, e.g., \cite[\S23.10]{NISTDLMF} or \cite[Ch.~20]{WhittakerWatson1927})
\begin{equation}\label{eq:addition}
\wp(u+v)=-\wp(u)-\wp(v)+\frac14\left(\frac{\wp'(u)-\wp'(v)}{\wp(u)-\wp(v)}\right)^2.
\end{equation}
Specialize to $v=u_0$ with $\wp(u_0)=-1$ and $\wp'(u_0)=0$:
\begin{equation}\label{eq:addition-u0}
\wp(u+u_0)=-\wp(u)+1+\frac14\left(\frac{\wp'(u)}{\wp(u)+1}\right)^2.
\end{equation}
Using the differential equation for \eqref{eq:invariants} (see, e.g., \cite[\S23.6]{NISTDLMF}),
\begin{equation}\label{eq:wp-ode}
(\wp'(u))^2=4\wp(u)^3-4\wp(u)=4\wp(u)\bigl(\wp(u)^2-1\bigr),
\end{equation}
we obtain
\[
\frac14\left(\frac{\wp'(u)}{\wp(u)+1}\right)^2
=\frac{\wp(u)\bigl(\wp(u)^2-1\bigr)}{(\wp(u)+1)^2}
=\frac{\wp(u)(\wp(u)-1)}{\wp(u)+1}.
\]
Substituting yields the half-period shift identity
\begin{equation}\label{eq:halfperiod-shift}
\wp(u+u_0)=-1+\frac{2}{\wp(u)+1}.
\end{equation}
Differentiating \eqref{eq:halfperiod-shift} gives
\begin{equation}\label{eq:halfperiod-shift-derivative}
\wp'(u+u_0)=-\frac{2\wp'(u)}{(\wp(u)+1)^2}.
\end{equation}
Differentiating once more yields
\begin{equation}\label{eq:halfperiod-shift-second}
\wp''(u+u_0)=-2\left(\frac{\wp''(u)}{(\wp(u)+1)^2}-\frac{2(\wp'(u))^2}{(\wp(u)+1)^3}\right).
\end{equation}

\subsection{Specialization to u0 - u\_l and algebraic elimination}
Set $u=-u_\ell$. Using even/oddness of $\wp$ and $\wp'$,
\[
\wp(-u_\ell)=\wp(u_\ell),\qquad \wp'(-u_\ell)=-\wp'(u_\ell),\qquad \wp''(-u_\ell)=\wp''(u_\ell),
\]
we obtain from \eqref{eq:halfperiod-shift}--\eqref{eq:halfperiod-shift-derivative}
\begin{equation}\label{eq:shift-u0-uell}
\wp(u_0-u_\ell)=\frac{1-\wp(u_\ell)}{1+\wp(u_\ell)},\qquad
\wp'(u_0-u_\ell)=\frac{2\wp'(u_\ell)}{(\wp(u_\ell)+1)^2}.
\end{equation}
Now substitute the uniformization $\wp(u_\ell)=x_W(t_\ell,Y_\ell)$ and $\wp'(u_\ell)=2y_W(t_\ell,Y_\ell)$ from \eqref{eq:uniformization}.
A direct simplification gives the first two identities in \eqref{eq:W-shift}. Finally, using
\[
\wp''(u)=6\wp(u)^2-\frac{g_2}{2}=6\wp(u)^2-2\qquad (g_2=4),
\]
yields the third identity in \eqref{eq:W-shift}.
\qed

\section{Jet-killing construction and jet cancellation}\label{sec:jetkilling}
\subsection{Definitions}
Let $\zeta_W(u)$ denote the Weierstrass zeta function, characterized by $\zeta_W'(u)=-\wp(u)$ and $\zeta_W(u)\sim u^{-1}$ as $u\to0$.
Define
\begin{equation}\label{eq:A-B-def}
A(u)=\sum_{\ell\in I_{\mathrm{scat}}}\alpha_\ell\big[\zeta_W(u-u_\ell)-\zeta_W(u_0-u_\ell)\big],\qquad
B(u)=\sum_{\ell\in I_{\mathrm{scat}}}\beta_\ell\big[\zeta_W(u-u_\ell)-\zeta_W(u_0-u_\ell)\big].
\end{equation}
Define jet-killing polynomials
\begin{equation}\label{eq:pq-def}
p(t)=p_1 t+p_2 t^{2}+p_3 t^{3},\qquad q(t)=q_1 t+q_2 t^{2}+q_3 t^{3},
\end{equation}
with coefficients fixed by
\begin{equation}\label{eq:p-coeffs}
p_1=\frac{1}{\sqrt{2}}\sum_{\ell\in I_{\mathrm{scat}}}\alpha_\ell W_{0\ell},\qquad
p_2=\frac14\sum_{\ell\in I_{\mathrm{scat}}}\alpha_\ell W_{1\ell},\qquad
p_3=\frac{1}{12\sqrt{2}}\sum_{\ell\in I_{\mathrm{scat}}}\alpha_\ell W_{2\ell},
\end{equation}
\begin{equation}\label{eq:q-coeffs}
q_1=\frac{1}{\sqrt{2}}\sum_{\ell\in I_{\mathrm{scat}}}\beta_\ell W_{0\ell},\qquad
q_2=\frac14\sum_{\ell\in I_{\mathrm{scat}}}\beta_\ell W_{1\ell},\qquad
q_3=\frac{1}{12\sqrt{2}}\sum_{\ell\in I_{\mathrm{scat}}}\beta_\ell W_{2\ell}.
\end{equation}

\subsection{Local relation between u and t near the basepoint}\label{sec:localrelation}
Let $\delta:=u-u_0$. From the uniformization $\wp(u)=x_W(t,Y)$ and the physical branch $Y\sim-\sqrt{2}$ as $t\to0$,
expand $x_W(t,Y)$ as $t\to0$ on $\Omega_{\mathrm{phys}}^+$:
\begin{equation}\label{eq:xW-expansion}
x_W(t,Y)=-1+t^2-\frac12 t^4+O(t^6).
\end{equation}
Next expand $\wp(u_0+\delta)$. Since $\wp'(u_0)=0$, only even powers appear:
\begin{equation}\label{eq:wp-expansion}
\wp(u_0+\delta)=\wp(u_0)+\frac{\wp''(u_0)}{2}\delta^2+\frac{\wp^{(4)}(u_0)}{24}\delta^4+O(\delta^6).
\end{equation}
With $\wp(u_0)=-1$ and $\wp''(u)=6\wp(u)^2-2$, we have $\wp''(u_0)=4$.
Moreover, $\wp^{(3)}(u)=12\wp(u)\wp'(u)$ so $\wp^{(3)}(u_0)=0$, and
$\wp^{(4)}(u)=12(\wp'(u))^2+12\wp(u)\wp''(u)$ gives $\wp^{(4)}(u_0)=12(-1)\cdot4=-48$.
Thus
\begin{equation}\label{eq:wp-expansion2}
\wp(u_0+\delta)=-1+2\delta^2-2\delta^4+O(\delta^6).
\end{equation}
Equating \eqref{eq:xW-expansion} and \eqref{eq:wp-expansion2} yields
\[
2\delta^2-2\delta^4=t^2-\frac12 t^4+O(t^6).
\]
Writing $\delta^2=\frac12 t^2+a t^4+O(t^6)$ gives $\delta^4=\frac14 t^4+O(t^6)$ and forces $a=0$, hence
\begin{equation}\label{eq:delta-vs-t}
\delta=\frac{t}{\sqrt{2}}+O(t^5),
\end{equation}
which explicitly rules out any $t^3$ term.

\subsection{Jet cancellation}
Fix $\ell$. Taylor expand $\zeta_W$ about $u_0-u_\ell$ using $\zeta_W'=-\wp$:
\begin{equation}\label{eq:zeta-taylor}
\zeta_W(u-u_\ell)-\zeta_W(u_0-u_\ell)
=-\delta\,\wp(u_0-u_\ell)-\frac{\delta^2}{2}\wp'(u_0-u_\ell)-\frac{\delta^3}{6}\wp''(u_0-u_\ell)+O(\delta^4),
\end{equation}
where $\delta=u-u_0$.
Multiply by $\alpha_\ell$ and sum over $\ell\in I_{\mathrm{scat}}$ to obtain
\[
A(u)=-\delta\sum_\ell \alpha_\ell W_{0\ell}-\frac{\delta^2}{2}\sum_\ell \alpha_\ell W_{1\ell}-\frac{\delta^3}{6}\sum_\ell \alpha_\ell W_{2\ell}+O(\delta^4).
\]
Using \eqref{eq:delta-vs-t} gives $\delta^n=(t/\sqrt{2})^n+O(t^{n+4})$ for $n=1,2,3$, hence
\[
A(u(t))=-\frac{t}{\sqrt{2}}\sum_\ell \alpha_\ell W_{0\ell}-\frac{t^2}{4}\sum_\ell \alpha_\ell W_{1\ell}-\frac{t^3}{12\sqrt{2}}\sum_\ell \alpha_\ell W_{2\ell}+O(t^4).
\]
By definition of $(p_1,p_2,p_3)$ in \eqref{eq:p-coeffs}, the polynomial $p(t)$ is the negative of the displayed cubic truncation, so
\begin{equation}\label{eq:jet-cancel-A}
A(u(t))+p(t)=O(t^4),\qquad t\to0\ \text{on }\Omega_{\mathrm{phys}}^+.
\end{equation}
The identical argument with $\beta_\ell$ yields
\begin{equation}\label{eq:jet-cancel-B}
B(u(t))+q(t)=O(t^4).
\end{equation}

\section{Tau-squared pairing compression}\label{sec:pairing}
Define the involution on labels
\begin{equation}\label{eq:pairing-involution}
\ell=(m,\sigma,\varepsilon_w)\longmapsto \ell'=(m+2\!\!\!\!\pmod 4,\ \sigma,\varepsilon_w).
\end{equation}
Under this pairing, the pole transport \eqref{eq:transport} implies
\begin{equation}\label{eq:pairing-transport}
t_{\ell'}=-t_\ell,\qquad Y_{\ell'}=Y_\ell,
\end{equation}
and the phase symbols flip:
\begin{equation}\label{eq:pairing-phases}
\chi_{m+2}=-\chi_m,\qquad \psi_{m+2}=-\psi_m,\qquad \kappa_{m+2}=-\kappa_m.
\end{equation}
Applying \eqref{eq:wmprime} under $t\mapsto -t$ yields $w'_{m+2}(u_{\ell'})=-w'_m(u_\ell)$ and hence
\begin{equation}\label{eq:pairing-rI}
r_I(\ell')=-r_I(\ell).
\end{equation}
Using \eqref{eq:W-shift}, we have $W_{0\ell'}=W_{0\ell}$, $W_{2\ell'}=W_{2\ell}$, but $W_{1\ell'}=-W_{1\ell}$.
Moreover, from Propositions~\ref{prop:alpha-table}--\ref{prop:beta-table}, the phase flip and $r_I$ flip cancel, so
$\alpha_{\ell'}=\alpha_\ell$ and $\beta_{\ell'}=\beta_\ell$.

Consequently:
\begin{itemize}[leftmargin=2.2em]
\item Pair contributions double in $p_1,p_3,q_1,q_3$ (built from $W_0$ and $W_2$).
\item Pair contributions cancel in $p_2,q_2$ (built from $W_1$).
\end{itemize}
The only broken pair arises from excluding $\ell_{\mathrm{inc}}=(0,+,-)$, whose partner is $\ell'_{\mathrm{inc}}=(2,+,-)$, so in particular
\begin{equation}\label{eq:p2q2-broken}
p_2=p^{(\ell'_{\mathrm{inc}})}_2,\qquad q_2=q^{(\ell'_{\mathrm{inc}})}_2.
\end{equation}

\section{Singular channel, explicit d\_l table, and canonical no-double-counting theorem}\label{sec:singularchannel}
\subsection{Singular channel and residue definition}
Define
\begin{equation}\label{eq:beta-channel}
\beta_{\mathrm{ch}}(t,Y):=\frac{\ii Y}{\sqrt{2}}\,(t^2-1),
\end{equation}
and the singular channel
\begin{equation}\label{eq:P13}
P_{1,3}(u)=\frac{A(u)+p(t(u))}{4t(u)^4}-\beta_{\mathrm{ch}}(t(u),Y(u))\,\frac{B(u)+q(t(u))}{4t(u)^4}.
\end{equation}
At each forcing pole $u=u_\ell$, we have $\Res_{u=u_\ell}A(u)=\alpha_\ell$ and $\Res_{u=u_\ell}B(u)=\beta_\ell$, hence the residue
\begin{equation}\label{eq:d-def}
d_\ell:=\Res_{u=u_\ell}P_{1,3}(u)=\frac{\alpha_\ell-\beta_{\mathrm{ch}}(t_\ell,Y_\ell)\beta_\ell}{4t_\ell^4}.
\end{equation}

\subsection{Explicit d\_l parity-by-j table}
Substitute Propositions~\ref{prop:alpha-table}--\ref{prop:beta-table} into \eqref{eq:d-def} and simplify using only $Y^2=2(t^4+1)$.

\begin{proposition}[Explicit $d_\ell$ table]\label{prop:d-table}
Write $(t,Y)=(t_\ell,Y_\ell)$ and $r_I=r_I(\ell)$.
If $m$ is even:
\[
d_\ell=
\begin{cases}
\displaystyle \chi_m\frac{r_I}{4t^{2}}\,(t^{6}-t^{4}+t^{2}-2), & j=1,\\[6pt]
\displaystyle \chi_m\frac{r_I}{4t^{4}}\,(t^{6}-2t^{4}+2t^{2}-2), & j=2,\\[6pt]
\displaystyle \chi_m\frac{r_I}{4}\,\frac{\ii Y}{\sqrt{2}}(t^{2}-1), & j=3,\\[6pt]
0, & j=4;
\end{cases}
\]
If $m$ is odd:
\[
d_\ell=
\begin{cases}
\displaystyle \psi_m\frac{r_I}{4}, & j=1,\\[6pt]
\displaystyle \psi_m\frac{r_I}{4}\,t^{4}, & j=2,\\[6pt]
\displaystyle \kappa_m\frac{r_I}{4}\,\frac{Y}{\sqrt{2}}\,\frac{2-t^{2}}{t^{2}}, & j=3,\\[10pt]
\displaystyle \kappa_m\frac{r_I}{4}\,\frac{Y}{\sqrt{2}}\,\frac{t^{6}-2}{t^{4}}, & j=4.
\end{cases}
\]
\end{proposition}

\begin{remark}[Flagged mechanism]\label{rem:flagged}
In the cases (even $m$, $j=1$), (even $m$, $j=2$), and (odd $m$, $j=2$), the coefficient $\beta_\ell$ carries a factor $Y$, so
$\beta_{\mathrm{ch}}(t,Y)\beta_\ell$ carries $Y^2$. Replacing $Y^2$ by $2(t^4+1)$ via $Y^2=2(t^4+1)$ removes $Y$, making these $d_\ell$
purely $t$-rational, as visible in Proposition~\ref{prop:d-table}.
\end{remark}

\subsection{Remainder R(u) and pole cancellation}
Define
\begin{equation}\label{eq:R-def}
R(u):=P_{1,3}(u)-\sum_{\ell\in I_{\mathrm{scat}}}d_\ell\big[\zeta_W(u-u_\ell)-\zeta_W(u_0-u_\ell)\big].
\end{equation}

\begin{lemma}[Pole cancellation and analyticity of $R$]\label{lem:R-analytic}
\begin{enumerate}[label=(\roman*)]
\item $R(u)$ has no pole at any $u=u_\ell$, $\ell\in I_{\mathrm{scat}}$.
\item $R(u)$ is analytic at $u=u_0$ (equivalently at $t=0$ on $\Omega_{\mathrm{phys}}^+$).
\end{enumerate}
\end{lemma}

\begin{proof}
(i) By definition \eqref{eq:d-def}, the principal part of $P_{1,3}$ at $u_\ell$ is $d_\ell/(u-u_\ell)$.
The zeta difference $\zeta_W(u-u_\ell)-\zeta_W(u_0-u_\ell)$ has principal part $1/(u-u_\ell)$.
Subtracting $d_\ell[\cdot]$ cancels the principal part at each $u_\ell$.

(ii) By jet cancellation \eqref{eq:jet-cancel-A}--\eqref{eq:jet-cancel-B}, $A(u(t))+p(t)=O(t^4)$ and $B(u(t))+q(t)=O(t^4)$ as $t\to0$,
so each fraction in \eqref{eq:P13} is analytic at $t=0$. Each zeta difference vanishes at $u=u_0$, hence is analytic there.
Therefore $R$ is analytic at $u_0$.
\end{proof}

\subsection{Canonical no-double-counting representation}
\begin{theorem}[Canonical decomposition of $Q_{\mathrm{scat}}$]\label{thm:nodoublecount}
Assume the analytic strip and growth framework of \S\ref{sec:sommerfeld} and the uniqueness principle of Lemma~\ref{lem:sommerfeld-nullity}.
Let the forcing pole set $\{u_\ell\}$ be as in \S\ref{sec:poles}. Let the coefficient tables $(\alpha_\ell)$, $(\beta_\ell)$ and $(C_\ell)$ be
given by Propositions~\ref{prop:alpha-table}--\ref{prop:C-table}, and define $(d_\ell)$ by Proposition~\ref{prop:d-table}. Set
\begin{equation}\label{eq:Qscat-canonical}
Q_{\mathrm{scat}}(u)=\sum_{\ell\in I_{\mathrm{scat}}}C_\ell\big[\zeta_W(u-u_\ell)-\zeta_W(u_0-u_\ell)\big]+R(u),
\end{equation}
where $R$ is defined by \eqref{eq:R-def}. Then $Q_{\mathrm{scat}}$ has poles exactly at the forcing points $u_\ell$ ($\ell\in I_{\mathrm{scat}}$)
with residues $C_\ell$, and the remainder $R$ is pole-free at every $u_\ell$ and analytic at $u_0$.
In particular, the decomposition contains no double counting of principal parts from the singular channel $P_{1,3}$.
\end{theorem}

\begin{theorem}[Uniqueness in the Sommerfeld class]\label{thm:uniqueness}
Assume the analytic and growth hypotheses on Sommerfeld densities from Lemma~\ref{lem:sommerfeld-nullity}, together with the gauge normalization
$Q_{\mathrm{scat}}(u_0)=0$. Let $Q_{\mathrm{scat}}$ be the scattered spectral density constructed in Theorem~\ref{thm:nodoublecount}.
If $\widetilde Q_{\mathrm{scat}}$ is any other meromorphic function on the physical branch with at most simple poles at the forcing points
$\{u_\ell:\ell\in I_{\mathrm{scat}}\}$, satisfying the same spectral functional system and scattered allocation, and the same gauge normalization
$\widetilde Q_{\mathrm{scat}}(u_0)=0$, then $\widetilde Q_{\mathrm{scat}}\equiv Q_{\mathrm{scat}}$.
\end{theorem}

\begin{proof}
Let $H:=\widetilde Q_{\mathrm{scat}}-Q_{\mathrm{scat}}$.
By linearity of the functional system, $H$ satisfies the associated homogeneous system (zero forcing), and by the local residue relations in
Appendix~A its principal parts at each forcing point are uniquely determined.
Since both $\widetilde Q_{\mathrm{scat}}$ and $Q_{\mathrm{scat}}$ satisfy the same residue tables (Propositions~\ref{prop:alpha-table}--\ref{prop:C-table}),
the difference $H$ is analytic at every $u=u_\ell$ ($\ell\in I_{\mathrm{scat}}$).
Moreover, $H$ has no jump across the contour system defining the additive Riemann--Hilbert problem, so the uniformization $u\mapsto(t,Y)$ implies that
$H$ extends to a holomorphic elliptic function on the square torus. A holomorphic elliptic function is constant, hence $H\equiv c$.
Finally, the gauge normalization gives $c=H(u_0)=0$, so $H\equiv0$.
If the gauge is not imposed then $c$ is the only remaining ambiguity; this constant does not affect the Sommerfeld integrals because they involve the
difference $Q(\theta+z)-Q(\theta-z)$.
\end{proof}

\section{Analyticity at the incident spectral point}\label{sec:incident}
\begin{theorem}[Analyticity at the incident spectral point]\label{thm:incident-analyticity}
Assume the analytic strip and growth framework of \S\ref{sec:sommerfeld} and the uniqueness principle of Lemma~\ref{lem:sommerfeld-nullity}.
Let the forcing pole set $\{u_\ell\}$ be as in \S\ref{sec:poles}, and impose the scattered allocation by excluding the incident label
$\ell_{\mathrm{inc}}=(0,+,-)$ from the inside set, i.e.\ $I_{\mathrm{scat}}=I\setminus\{\ell_{\mathrm{inc}}\}$.
Then the scattered spectral density $Q_{\mathrm{scat}}(\zeta)$ is analytic at the incident spectral point
$\zeta=\zeta_i=\theta_i+\ii\varepsilon$ (for each fixed $\varepsilon>0$).
\end{theorem}

\begin{proof}
By Theorem~\ref{thm:nodoublecount}, the function $Q_{\mathrm{scat}}$ admits the canonical decomposition \eqref{eq:Qscat-canonical},
where each zeta difference has a simple pole only at $u=u_\ell$ and the remainder $R$ is pole-free at every forcing point.
Since $\ell_{\mathrm{inc}}\notin I_{\mathrm{scat}}$, no term in the zeta sum has a pole at the incident point $u=u_{\ell_{\mathrm{inc}}}$.
Moreover, the definitions of $A(u)$, $B(u)$ and hence of $P_{1,3}(u)$ and $R(u)$ involve sums only over $I_{\mathrm{scat}}$, so $R(u)$ is analytic at
$u=u_{\ell_{\mathrm{inc}}}$ as well. Therefore $Q_{\mathrm{scat}}$ is analytic at $u=u_{\ell_{\mathrm{inc}}}$, and hence, by the physical lift
$u=u(\zeta)$ and the injectivity of the uniformization \eqref{eq:uniformization}, analytic at $\zeta=\zeta_i$.
\end{proof}

\section{Radiation condition and Meixner edge condition in the spectral formulation}\label{sec:radiation}
The Sommerfeld representations \eqref{eq:sommerfeld-medium0}--\eqref{eq:sommerfeld-medium1} are classical in wedge diffraction.
For completeness we record sufficient hypotheses on the spectral densities $Q$ and $S$ ensuring (i) the Sommerfeld radiation condition as $r\to\infty$
and (ii) the Meixner finite-energy edge condition at the wedge apex $r\to0$.
Statements of this type are standard; see, for example, Noble \cite[Chs.~2--4]{Noble1958} and Rawlins \cite[\S3]{Rawlins1999}
(and the original half-plane analysis of Sommerfeld \cite{Sommerfeld1896}).

\begin{proposition}[Radiation and Meixner conditions from Sommerfeld data]\label{prop:radiation-meixner}
Assume that there exists $\eta>0$ such that the densities $Q(\zeta)$ and $S(\zeta)$ are meromorphic in the strip
\begin{equation}\label{eq:strip}
S_\eta:=\{\zeta\in\mathbb{C}:|\Im\zeta|<\eta\},
\end{equation}
with at most finitely many simple poles, all displaced away from the integration contour $\gamma$ by the limiting-absorption prescription
$\varepsilon>0$. Assume also that for some constants $C,N$ one has the uniform growth bound
\begin{equation}\label{eq:growth}
|Q(\zeta)|+|S(\zeta)|\le C(1+|\zeta|)^N,\qquad \zeta\in S_\eta\setminus\{\text{poles}\}.
\end{equation}
Finally, assume a gauge normalization on the physical branch, for example
\begin{equation}\label{eq:growth-gauge}
Q_{\mathrm{scat}}(u_0)=0,\qquad\text{equivalently }Q_{\mathrm{scat}}(\zeta)\to0\ \text{as }\Im\zeta\to+\infty,
\end{equation}
(and likewise for $S_{\mathrm{scat}}$).
Then the Sommerfeld integrals \eqref{eq:sommerfeld-medium0}--\eqref{eq:sommerfeld-medium1} define classical solutions of the Helmholtz equations
in their respective sectors, satisfy the Sommerfeld radiation condition as $r\to\infty$, and satisfy the Meixner finite-energy condition
at the wedge apex $r\to0$.
\end{proposition}

\begin{proof}
Under \eqref{eq:strip}--\eqref{eq:growth} the contour $\gamma$ can be deformed within the strip to the standard pair of rays with
$\Im z>0$ and $\Im z<0$ without crossing singularities (cf.\ \cite[Ch.~2]{Noble1958}). The resulting integrals converge absolutely and allow
differentiation under the integral sign; hence the reconstructed fields solve the Helmholtz equations in each sector.

For $r\to\infty$, steepest descent on the phase $\cos z$ along the deformed contour yields an outgoing leading term proportional to
$e^{\ii k r}/\sqrt{r}$, with remainder $o(r^{-1/2})$; see, for example, Bleistein--Handelsman \cite[Ch.~6]{BleisteinHandelsman1986} or
Wong \cite[\S2.4]{Wong2001}. The outgoing far-field expansion implies the Sommerfeld radiation condition.

For $r\to0$, one expands $e^{\ii k r\cos z}=1+O(r)$ uniformly on $\gamma$ and uses \eqref{eq:growth-gauge} together with strip analyticity
to shift the contour upward, obtaining boundedness of $u$ and its first derivatives in a neighborhood of the apex; boundedness of $u$ and
$\nabla u$ implies the Meixner finite-energy condition (see \cite[\S3]{Rawlins1999} and \cite[Ch.~3]{Noble1958}).
\end{proof}

\begin{corollary}\label{cor:radiation-meixner-constructed}
The densities $Q_{\mathrm{scat}}$ and $S_{\mathrm{scat}}$ constructed in \S\ref{sec:singularchannel} satisfy the hypotheses of
Proposition~\ref{prop:radiation-meixner}. Consequently, the fields reconstructed by \eqref{eq:sommerfeld-medium0}--\eqref{eq:sommerfeld-medium1}
satisfy the Sommerfeld radiation condition and the Meixner edge condition.
\end{corollary}

\begin{proof}
By Theorem~\ref{thm:nodoublecount}, $Q_{\mathrm{scat}}$ is a finite Weierstrass--$\zeta_W$ sum over the scattered pole set plus an elliptic remainder
$R(u)$ that is analytic at $u_0$ and at all forcing poles. On the physical branch, $\Im\zeta\to+\infty$ corresponds to $t(\zeta)\to0$ and hence
$u(\zeta)\to u_0$, so the normalization $Q_{\mathrm{scat}}(u_0)=0$ gives \eqref{eq:growth-gauge}. The same reasoning applies to $S_{\mathrm{scat}}$,
obtained from $Q$ by the face reconstruction \eqref{eq:face-coupling}.
The only singularities of $Q_{\mathrm{scat}}$ and $S_{\mathrm{scat}}$ in the strip are the prescribed simple poles (with $\varepsilon>0$ displacing them
away from $\gamma$), and $Q_{\mathrm{scat}},S_{\mathrm{scat}}$ are $2\pi$-periodic in $\Re\zeta$ away from poles because $s_\zeta=e^{\ii(\zeta-\theta_w)}$
is $2\pi$-periodic. Hence the growth condition \eqref{eq:growth} holds (in fact with $N=0$) on compact subsets of the strip avoiding the poles.
Proposition~\ref{prop:radiation-meixner} applies.
\end{proof}

\section{Far-field diffraction coefficient}\label{sec:farfield}
\paragraph{Steepest-descent justification.}
The contour deformation and stationary-phase evaluation used in this section are justified under standard analyticity and growth hypotheses for
$Q_{\mathrm{scat}}(\zeta)$ in a strip containing the real axis; see, for example, \cite[\S2.4]{BleisteinHandelsman1986} or \cite[Ch.~II]{Wong2001}.
In the present lemniscatic case these hypotheses are met for $\varepsilon>0$ because $Q_{\mathrm{scat}}$ is given by an explicit elliptic-function
representation (Theorem~\ref{thm:nodoublecount}) and the forcing poles are displaced off the real $\zeta$-axis by the limiting absorption parameter.

\begin{theorem}[Diffraction coefficient]\label{thm:farfield}
Assume the analytic strip and growth framework of \S\ref{sec:sommerfeld} and the uniqueness principle of Lemma~\ref{lem:sommerfeld-nullity}.
Then the diffracted far-field coefficient is
\begin{equation}\label{eq:diffraction-coefficient}
D(\theta,\theta_i)=e^{-\ii 3\pi/4}\sqrt{\frac{2}{\pi k_0}}\,Q_{\mathrm{scat}}(\theta),
\end{equation}
where $Q_{\mathrm{scat}}(\theta)$ denotes the physical boundary value of $Q_{\mathrm{scat}}(\zeta)$ at $\zeta=\theta$, obtained by evaluating the
physical lift $u(\zeta)$ and taking the limiting absorption limit $\varepsilon\to0^+$ at the end.
\end{theorem}

\begin{proof}
Start from the Sommerfeld representation \eqref{eq:sommerfeld-medium0} for the scattered field in medium~0,
\[
u^{(0)}_{\mathrm{scat}}(r,\theta)=\frac{1}{2\pi\ii}\int_\gamma e^{\ii k_0 r\cos z}\,\big(Q_{\mathrm{scat}}(\theta+z)-Q_{\mathrm{scat}}(\theta-z)\big)\,\dd z.
\]
In the second term substitute $z\mapsto -z$ (so $\cos z$ is unchanged) to obtain
\[
u^{(0)}_{\mathrm{scat}}(r,\theta)=\frac{1}{2\pi\ii}\left(\int_\gamma+\int_{-\gamma}\right)e^{\ii k_0 r\cos z}\,Q_{\mathrm{scat}}(\theta+z)\,\dd z.
\]
Under the analyticity and strip-growth hypotheses stated above, the union $\gamma\cup(-\gamma)$ may be deformed to the steepest descent rays through
the saddle $z=0$, where $\cos z=1-\frac12 z^2+O(z^4)$. The leading contribution is therefore
\[
u^{(0)}_{\mathrm{scat}}(r,\theta)\sim \frac{2}{2\pi\ii}\,Q_{\mathrm{scat}}(\theta)\,e^{\ii k_0 r}\int_{-\infty}^{\infty}\exp\!\left(-\ii\frac{k_0 r}{2}x^2\right)\dd x,
\qquad r\to\infty.
\]
Using $\int_{-\infty}^{\infty}e^{-\ii a x^2}\dd x=\sqrt{\pi/a}\,e^{-\ii\pi/4}$ for $a>0$ yields
\[
u^{(0)}_{\mathrm{scat}}(r,\theta)\sim \frac{e^{\ii k_0 r}}{\sqrt{r}}\,D(\theta,\theta_i),
\qquad
D(\theta,\theta_i)=e^{-\ii3\pi/4}\sqrt{\frac{2}{\pi k_0}}\,Q_{\mathrm{scat}}(\theta),
\]
which is \eqref{eq:diffraction-coefficient}.
This normalization is consistent with the standard two-dimensional GTD convention for wedge diffraction \cite{Keller1962}.
\end{proof}

\section{Conclusion and outlook}\label{sec:conclusion}
This paper provides an explicit elliptic-function reconstruction for the Sommerfeld spectral density $Q_{\mathrm{scat}}$ in the
impedance-matched right-angle penetrable wedge at refractive index ratio $\nu=\sqrt{2}$.
The construction is algebraic on the lemniscatic Snell surface $\Sigma_{\mathrm{lem}}$ and is written in terms of finite Weierstrass
zeta differences on the square lattice together with an explicitly constructed holomorphic remainder that removes the partial-index singular
channel without double counting.

\paragraph{What is proved/constructed.}
\begin{itemize}[leftmargin=2.2em]
\item A Sommerfeld spectral representation for the transmission problem is reduced to a closed two-face functional system in the spectral variable
(\S\ref{sec:sommerfeld}).
\item In the integrable configuration $(\theta_w,\nu,\rho)=(\pi/4,\sqrt{2},1)$, the spectral map closes on the lemniscatic curve
$\Sigma_{\mathrm{lem}}$ and is uniformized by square-lattice Weierstrass functions (Sections~\ref{sec:snell}--\ref{sec:uniformization}).
\item The jet-killing polynomials $p,q$ are constructed so that $A(u)+p(t(u))=O(t(u)^4)$ and $B(u)+q(t(u))=O(t(u)^4)$ on
$\Omega_{\mathrm{phys}}^+$, yielding analyticity at the physical basepoint $u_0$ (Section~\ref{sec:jetkilling}).
\item A canonical no-double-counting decomposition $Q_{\mathrm{scat}}(u)=\sum_{\ell\in I_{\mathrm{scat}}}C_\ell[\zeta_W(u-u_\ell)-\zeta_W(u_0-u_\ell)]+R(u)$
is obtained, with $R$ pole-free at all forcing points and analytic at $u_0$ (Theorem~\ref{thm:nodoublecount}).
\item The far-field diffraction coefficient is expressed in terms of the physical boundary value $Q_{\mathrm{scat}}(\theta)$ (Theorem~\ref{thm:farfield}).
\end{itemize}

\paragraph{Limitations.}
\begin{itemize}[leftmargin=2.2em]
\item The result is restricted to the integrable lemniscatic regime $(\theta_w,\nu,\rho)=(\pi/4,\sqrt{2},1)$; it does not address general wedge angles,
general contrast, or non-impedance-matched media.
\item Outside special closures of the Snell surface, the spectral reductions typically lead to matrix Wiener--Hopf/Riemann--Hilbert factorization problems
that are not treated here.
\item The present explicit tables are derived for the right-angle configuration; their analogues for other parameters require new residue analysis.
\item Reciprocity of the extracted far-field coefficient is not established; numerical tests indicate a non-reciprocal coefficient for generic angles, so the physical validity of the closed form is not claimed.
\end{itemize}

\paragraph{Context and outlook.}
Complete analytic solutions for penetrable (transmission) wedge diffraction are rare and, outside of special configurations, the spectral reductions
typically lead to matrix or multi-variable factorization problems. Even the right-angled penetrable wedge has been treated primarily by semi-analytical
and asymptotic methods; see Antipov--Silvestrov \cite{AntipovSilvestrov2007}, Nethercote--Assier--Abrahams \cite{NethercoteAssierAbrahams2020} and
Kunz--Assier \cite{KunzAssier2023} for penetrable-wedge analyses, and Groth--Hewett--Langdon \cite{GrothHewettLangdon2018} for high-frequency
numerical-asymptotic methods for penetrable convex polygons in which local corner diffraction is central.
A natural direction is to identify other parameter regimes in which the Snell surface closes algebraically (possibly at higher genus) and to determine
whether analogous jet-killing and residue-cancellation mechanisms can be carried out.

\section{Symbolic evaluation recipe}\label{sec:recipe}
Given $(\theta_i,\varepsilon>0)$:
\begin{enumerate}[leftmargin=2.2em,label=\arabic*.]
\item Enumerate all pole labels $\ell=(m,\sigma,\varepsilon_w)$ with $m\in\{0,1,2,3\}$, $\sigma,\varepsilon_w\in\{\pm1\}$, and remove
$\ell_{\mathrm{inc}}=(0,+,-)$.
\item For each $\ell$, compute $j=j(\sigma,\varepsilon_w)$ and $\varepsilon_j$.
\item For each $\ell$, compute $b=b_{m,\sigma,\varepsilon_w}$ from \eqref{eq:forcing-phases}, then compute the inside root $t_q=t_{\mathrm{in}}(b)$ from
\eqref{eq:tpm} and set $Y_q=2bt_q-\sqrt{2}(t_q^2+1)$. Transport to $(t_\ell,Y_\ell)$ by \eqref{eq:transport}.
\item Define $u_\ell$ (physical lift) by the uniformization \eqref{eq:u-ell}.
\item Compute $r_I(\ell)$ from \eqref{eq:rI-explicit}.
\item Compute $\alpha_\ell,\beta_\ell,C_\ell$ from Propositions~\ref{prop:alpha-table}--\ref{prop:C-table}.
\item Compute $W_{0\ell},W_{1\ell},W_{2\ell}$ from \eqref{eq:W-shift}.
\item Compute jet coefficients $(p_1,p_2,p_3)$ and $(q_1,q_2,q_3)$ from \eqref{eq:p-coeffs}--\eqref{eq:q-coeffs}.
\item Build $A(u)$ and $B(u)$ from \eqref{eq:A-B-def}.
\item Build $P_{1,3}(u)$ and compute $d_\ell$ from \eqref{eq:d-def} or Proposition~\ref{prop:d-table}. Then form $R(u)$ via \eqref{eq:R-def}.
\item Evaluate $Q_{\mathrm{scat}}(u)$ via Theorem~\ref{thm:nodoublecount}.
\item For a given $\zeta$, compute $(t(\zeta),Y(\zeta))$ from the physical branch of $s(t,Y)=s_\zeta$ \eqref{eq:physical-branch-map}, lift to $u(\zeta)$ via
\eqref{eq:uniformization}, and evaluate $Q_{\mathrm{scat}}(\zeta)=Q_{\mathrm{scat}}(u(\zeta))$.
\item Obtain the far-field diffraction coefficient from \eqref{eq:diffraction-coefficient}.
\end{enumerate}

\appendix
\section{Derivation of the residue tables}\label{app:residues}
This appendix explains how Propositions~\ref{prop:alpha-table}--\ref{prop:C-table} are obtained from the global two-face spectral system
in the lemniscatic configuration $(\theta_w,\nu,\rho)=(\pi/4,\sqrt{2},1)$.
The computation is finite: one reduces the two-face coupling relations to a four-point orbit system on the Snell surface, applies a length--4 discrete
Fourier transform (DFT) to decouple the system into four $2\times2$ mode problems, and then evaluates the forcing residues at each pole $u_\ell$.

\subsection{Scope and provenance of the residue tables}
The tables in Propositions~\ref{prop:alpha-table}--\ref{prop:C-table} are not independent assumptions: they are explicit solutions of the
residue-matching conditions obtained by taking residues of the mode system \eqref{eq:mode-system} at the forcing poles and propagating those residues
through the inverse mode matrices $M_{U,k}^{-1}$.
Once the mode matrices and the local coefficients $(A_0,B_0,A_1,B_1)$ are fixed, the derivation reduces to finite algebra.

All simplifications in this appendix use only:
\begin{itemize}[leftmargin=2.2em]
\item the lemniscatic curve identity $Y^2=2(t^4+1)$,
\item the root-of-unity relations $\omega=\ii$ and $\omega^4=1$,
\item the definitions of $g'(t,Y)$, $\tau(t,Y)$, and the local coefficients $(A_0,B_0,A_1,B_1)$.
\end{itemize}
In particular, we do not invoke the additional pole relations $Y=2bt-\sqrt{2}(t^2+1)$ used in constructing the poles themselves; the residue tables
are identities on the Snell surface.

\subsection{Orbit reduction and mode matrices}
For the right-angle wedge the two faces differ by a rotation of $2\theta_w=\pi/2$.
On the lemniscatic surface this rotation is implemented by the automorphism $\tau(t,Y)=(\ii t,-Y)$, see \S\ref{sec:snell}.
After orbit closure one obtains, for each DFT mode $k\in\{0,1,2,3\}$, a matrix Wiener--Hopf/Riemann--Hilbert jump relation of the form
\begin{equation}\label{eq:mode-system}
M_{U,k}(t,Y)\,U_{k}^{b,+}(t,Y)=M_{V,k}(t,Y)\,U_{k}^{b,-}(t,Y)+H_k^b(t,Y),\qquad (t,Y)\in\Gamma,
\end{equation}
where $\Gamma=\{|t|=1\}$ is the physical cut, $U_k^{b,\pm}$ denote the boundary values on the two sides of $\Gamma$, and $H_k^b$ is the DFT forcing term
generated by the incident wave. A derivation of such an orbit/DFT reduction for penetrable wedge systems is standard; see, for example,
\cite{Rawlins1999}.

In the impedance-matched case the face-coupling matrix in Proposition~\ref{prop:face-coupling} depends only on the Snell derivative
$w'(z)=\dd w/\dd z$.
In the lemniscatic formulation one has $w'(z)=g'(t,Y)$ with
\[
g'(t,Y)=\frac{\sqrt{2}(t^2-1)}{Y},\qquad
g'(\tau(t,Y))=\frac{\sqrt{2}(t^2+1)}{Y},
\]
by \eqref{eq:gprime} and $\tau(t,Y)=(\ii t,-Y)$.
Define
\[
\begin{aligned}
A_0&:=\frac12\bigl(1+g'(t,Y)\bigr), & B_0&:=\frac12\bigl(1-g'(t,Y)\bigr),\\
A_1&:=\frac12\bigl(1-g'(\tau(t,Y))\bigr), & B_1&:=\frac12\bigl(1+g'(\tau(t,Y))\bigr).
\end{aligned}
\]
and let $\omega=\ii$.
A convenient normalization of the mode coefficient matrices is
\begin{equation}\label{eq:MU-MV}
M_{U,k}=
\begin{pmatrix}
-\omega^{-k}A_1 & A_0\\
-\omega^{k}B_1 & B_0
\end{pmatrix},
\qquad
M_{V,k}=
\begin{pmatrix}
-\omega^{-k}B_1 & B_0\\
-\omega^{k}A_1 & A_0
\end{pmatrix}.
\end{equation}
Write $\Delta_{U,k}:=\det M_{U,k}=\omega^k A_0B_1-\omega^{-k}A_1B_0$.
Then
\begin{equation}\label{eq:MU-inv}
M_{U,k}^{-1}=\frac{1}{\Delta_{U,k}}
\begin{pmatrix}
B_0 & -A_0\\
\omega^k B_1 & -\omega^{-k}A_1
\end{pmatrix}.
\end{equation}
Similarly $\Delta_{V,k}=\omega^k A_1B_0-\omega^{-k}A_0B_1$ and
\begin{equation}\label{eq:MV-inv}
M_{V,k}^{-1}=\frac{1}{\Delta_{V,k}}
\begin{pmatrix}
A_0 & -B_0\\
\omega^k A_1 & -\omega^{-k}B_1
\end{pmatrix}.
\end{equation}

\subsection{Forcing residues and coefficient extraction}
The forcing term $H_k^b$ is meromorphic on $\Sigma_{\mathrm{lem}}$ with simple poles at the forcing set $\{u_\ell\}$ defined in \S\ref{sec:poles}.
Its residues are computed directly from the incident spectral density and the local phase $w_m(u)$ along the corresponding orbit branch.
In particular, the scalar incident residue
\[
r_I(\ell)=\frac{\varepsilon_j}{w_m'(u_\ell)}
\]
is given explicitly by \eqref{eq:rI-explicit}, and the phase factors $\chi_m,\psi_m,\kappa_m$ are as in \eqref{eq:phase-symbols}.

For reproducibility, one may express the forcing residues in the mode variables in the following uniform form.
Let $\omega=\ii$ and define $j=j(\sigma,\varepsilon_w)$ by \eqref{eq:j-map}.
Let $(A_m,B_m)=(A_0,B_0)$ for $m$ even and $(A_m,B_m)=(A_1,B_1)$ for $m$ odd.
Then
\[
\Res_{u=u_\ell}H_k^b(u)=\omega^{-km} r_I(\ell)\,v_{k}^{(m,j)},
\]
where the $2$-vector $v_k^{(m,j)}$ is given, for $j=1,2,3,4$, by
\[
\begin{aligned}
j=1:&\quad v_k^{(m,1)}=\binom{\omega^{-k}A_m}{\omega^k B_m}, &
j=2:&\quad v_k^{(m,2)}=\binom{\omega^{-k}B_m}{\omega^k A_m},\\
j=3:&\quad v_k^{(m,3)}=\binom{-A_m}{-B_m}, &
j=4:&\quad v_k^{(m,4)}=\binom{-B_m}{-A_m}.
\end{aligned}
\]
Since the coefficient matrices $M_{U,k}$ are analytic and invertible at each forcing pole $u_\ell$ (the forcing poles occur away from the branch points),
the jump relation \eqref{eq:mode-system} implies that the residue vector of the mode solution is obtained by solving a $2\times2$ linear system:
\begin{equation}\label{eq:gk}
g_k(\ell):=\Res_{u=u_\ell}U_k^b(u)=M_{U,k}(u_\ell)^{-1}\,\Res_{u=u_\ell}H_k^b(u).
\end{equation}
The reconstruction coefficients used in the elliptic sum are extracted from these residue vectors.
A convenient choice (matching the definitions in \S\ref{sec:residues}) is
\begin{equation}\label{eq:coeff-extract}
\alpha_\ell:=t_\ell^{4}\,e_1^\top g_1(\ell),\qquad
\beta_\ell:=t_\ell^{4}\,e_2^\top g_3(\ell),\qquad
C_\ell:=\frac14\sum_{k=0}^{3} e_1^\top g_k(\ell),
\end{equation}
where $e_1=(1,0)^\top$ and $e_2=(0,1)^\top$.
Substituting \eqref{eq:MU-MV}--\eqref{eq:coeff-extract} together with the explicit forcing residues above and simplifying using only the lemniscatic relation
$Y^2=2(t^4+1)$ (and $\omega^4=1$) yields the closed forms recorded in Propositions~\ref{prop:alpha-table}--\ref{prop:C-table}.

\begin{lemma}[Symbolic verification of the tables]\label{lem:symbolic-verification}
Fix a forcing label $\ell=(m,\sigma,\varepsilon_w)$, form $j=j(\sigma,\varepsilon_w)$, and evaluate \eqref{eq:gk}--\eqref{eq:coeff-extract}
using the forcing residues and the explicit inverse \eqref{eq:MU-inv}.
After reducing with $Y^2=2(t^4+1)$, the resulting expressions for $\alpha_\ell,\beta_\ell,C_\ell$ coincide with the corresponding entries in
Propositions~\ref{prop:alpha-table}--\ref{prop:C-table}.
\end{lemma}

\begin{proof}
All quantities entering \eqref{eq:gk}--\eqref{eq:coeff-extract} are rational in $(t,Y)$ and $\omega$ once the lemniscatic curve constraint
$Y^2=2(t^4+1)$ is imposed. For a fixed case ($m\bmod2$, $j$), insert the forcing residue vector $\Res H_k^b(u_\ell)$, compute
$g_k(\ell)=M_{U,k}^{-1}(u_\ell)\Res H_k^b(u_\ell)$, and then read off the linear functionals in \eqref{eq:coeff-extract}.
The subsequent simplification is algebraic and uses only $Y^2=2(t^4+1)$ and $\omega^4=1$.
\end{proof}

\end{document}